\documentclass[letterpaper, 10 pt, conference]{ieeeconf}  

\IEEEoverridecommandlockouts                              
\usepackage{cite}
\usepackage{graphicx}
\usepackage{amsmath}
\usepackage{amssymb}
\usepackage{booktabs}
\usepackage{multirow}
\allowdisplaybreaks[4]
\usepackage{mathrsfs}
\usepackage{verbatim}
\usepackage{amsfonts}
\usepackage{float}
\usepackage{subfigure}
\usepackage{caption}
\usepackage{subcaption}
\usepackage{hhline}
\usepackage{delarray}
\usepackage{epstopdf}
\usepackage{setspace}
\usepackage{enumerate}
\usepackage{lineno}
\usepackage{hyperref}
\usepackage{algorithm}
\usepackage{algorithmic}

\usepackage{xcolor}
\usepackage{color}

\usepackage{cases}
\usepackage{cancel}
\newtheorem{theorem}{\bf Theorem}

\newtheorem{definition}{\bf Definition}
\newtheorem{remark}{\bf Remark}
\newtheorem{assumption}{\bf Assumption}

\newtheorem{problem}{\bf Problem}

\title{\LARGE \bf
Distributed Safety-Critical MPC for Multi-Agent Formation Control  and Obstacle Avoidance}

\author{Chao Wang\textsuperscript{1}, Shuyuan Zhang\textsuperscript{2}, Lei Wang\textsuperscript{1}
\thanks{\textsuperscript{1}C. Wang and L. Wang are with the School of Automation Science and Electrical Engineering, Beihang University, Beijing, 100191, China 
        {\tt\small \{wchao, lwang\}@buaa.edu.cn}}
\thanks{\textsuperscript{2}S. Zhang is with the ICTEAM Institute, UCLouvain, 4 Avenue Georges Lema\^itre, 1348 Louvain-la-Neuve, Belgium
        {\tt\small shuyuan.zhang@uclouvain.be}}
}

\begin{document}

\maketitle
\thispagestyle{empty}
\pagestyle{empty}

\begin{abstract}
For nonlinear multi-agent systems with high relative degrees, achieving formation control and obstacle avoidance in a distributed manner remains a significant challenge. To address this issue, we propose a novel distributed safety-critical model predictive control (DSMPC) algorithm that incorporates discrete-time high-order control barrier functions (DHCBFs) to enforce safety constraints, alongside discrete-time control Lyapunov functions (DCLFs) to establish terminal constraints. To facilitate distributed implementation, we develop estimated neighbor states for formulating DHCBFs and DCLFs, while also devising a compatibility constraint to limit estimation errors and ensure convergence.  Additionally, we provide theoretical guarantees regarding the feasibility and stability of the proposed DSMPC algorithm based on a mild assumption. The effectiveness of the proposed method is evidenced by the simulation results, demonstrating improved performance and reduced computation time compared to existing approaches.
\end{abstract}

\section{INTRODUCTION}
\label{1}
Formation control and obstacle avoidance are fundamental capabilities in multi-agent systems (MASs), especially those with high relative degrees, such as quadrotor swarms and autonomous robotic convoys \cite{8798870}. These capabilities are essential for enabling coordinated behavior among agents while ensuring safety in complex environments.

Traditional methods like artificial potential fields \cite{9538804} and consensus-based approaches \cite{WU2020106332} are effective for formation control but lack predictive capability, making them less suitable for systems with high relative degrees and control constraints. Model predictive control (MPC) overcomes these limitations by explicitly incorporating system dynamics and constraints \cite{10384062, 2023centralized, VARGAS2022105054}. Centralized MPC frameworks \cite{2023centralized} have been applied to formation control and obstacle avoidance, while distributed MPC (DMPC) methods \cite{VARGAS2022105054} improve scalability. Despite their advantages, these approaches often rely on Euclidean distance constraints, which may be short-sighted in proactive obstacle avoidance.

Control barrier functions (CBFs) have been widely adopted to impose safety constraints by regulating approach speeds toward obstacles, thereby enabling more proactive avoidance strategies \cite{ 7782377, COHEN2024100947, 7857061, 7040372, 10167791}. In contrast to traditional Euclidean distance constraints, CBF-based methods ensure safety through real-time constraint satisfaction, often formulated within quadratic programming (QP) problems alongside nominal controllers or control Lyapunov functions (CLFs) to guarantee stability \cite{7782377, COHEN2024100947}. For MASs, CBFs have been extensively utilized to ensure safety while maintaining distributed coordination via QP-based approaches. Specifically, the study in \cite{7857061} developed real-time QP-based safety controllers for multi-robot systems to achieve collision-free operation, whereas \cite{9642050} optimized QP-based control signals to enforce safety constraints. Furthermore, the integration of CBFs with MPC enables long-horizon planning while preserving real-time safety guarantees. Research has demonstrated that MPC-based methods offer superior obstacle avoidance and stability performance compared to QP-based approaches \cite{9483029}. Additionally, it has been shown that appropriate CBF parameters can enhance MPC feasibility\cite{9683174}. However, these CBF-based methodologies are primarily designed for systems with a relative degree of one.

Recent studies have explored the application of high-order CBFs (HCBFs) to systems with high relative degrees \cite{9516971, 10886244, 10155975}. The concept of high-order CBFs was initially introduced in \cite{9516971}, with subsequent research extending it to robust formulations \cite{9868125}. By integrating high-order CBFs with MPC, improved control performance can be achieved for high-order systems. For instance, in \cite{10156532}, an iterative MPC framework utilizing DHCBFs was proposed for the safety-critical control of single-agent systems. However, the literature currently lacks comprehensive studies on the coordinated control of MASs with high relative degrees based on MPC and DHCBFs. For nonlinear MASs with high relative degrees, the challenge of achieving distributed implementation while ensuring recursive feasibility within the integrated MPC-CBF framework remains significant.
Motivated by these challenges, we aim to achieve formation control and obstacle avoidance for MASs with high relative degrees using a DMPC framework, while also ensuring the feasibility of the proposed approach. 

The main contributions of this paper are as follows: 
We propose a novel DSMPC algorithm for nonlinear MASs with high relative degrees. To formulate distributed safety and terminal constraints, DHCBFs and DCLFs are constructed using estimated future states of neighboring agents, obtained via a consensus-based controller. A compatibility constraint is introduced to limit estimation errors, ensuring constraints remain independent of control inputs of neighbors while enabling fully distributed optimization and guaranteeing convergence.
Unlike existing methods \cite{7857061, 10167791, 10156532}, which centrally handle MASs with relative degree one, our approach extends to nonlinear MASs with high relative degrees in a distributed manner, significantly enhancing computational efficiency. Moreover, under a mild assumption that the high-order time derivative of DHCBFs has a non-negative lower bound, we establish theoretical guarantees for recursive feasibility and stability. In contrast to MPC-CBF methods \cite{10167791, 10156532}, our approach explicitly ensures both feasibility and stability. 
Finally, the effectiveness of the proposed approach is validated through simulation results. Compared with existing methods, e.g., \cite{7857061, 2023centralized, 7040372}, our DSMPC algorithm achieves improved control performance with lower computational time.

The remainder of this paper is structured as follows. Section \ref{2} introduces key definitions and formulates the problem of interest. In Section \ref{3}, the DSMPC algorithm is presented, and the feasibility and stability are established. Simulation results are given in Section \ref{5}. Finally, Section \ref{6} concludes this paper.

This paper employs the following notations. The set of natural numbers is represented by $\mathbb{N}$, and the set of real numbers by $\mathbb{R}$. The $n$-dimensional real vector space is represented by $\mathbb{R}^{n}$, while $\mathbb{R}^{n \times n}$ denotes the space of $n \times n$ real matrices. $\mathcal{I}_{n}$ represents the $n \times n$ identity matrix. For a matrix $\Phi$, $\Phi \succ 0$ indicates that $\Phi$ is positive definite and symmetric. The Euclidean norm of a vector $x \in \mathbb{R}^{n}$ is denoted by $||x||$, and the weighted squared norm is defined as $||x||^{2}_{\Phi} = x^{\top} \Phi x$. $x \leq y$ means that each element of the vector $x$ is less than the corresponding element of the vector $y$. $\mathbf{0}_{n}$ represents the zero vector of dimension $n$.

\section{Preliminaries and Problem Formulation}
\label{2}

To depict the communication structure among MASs, graph theory is utilized as a foundational method. The interaction pattern among agents is modeled by a directed graph $\mathcal{G} = (\mathcal{V}, \mathcal{E}, \mathcal{W})$, where $\mathcal{V} = \{v_{1}, v_{2}, \ldots, v_{N}\}$ signifies the collection of $N$ agents, and  $\mathcal{E} \subseteq \mathcal{V} \times \mathcal{V}$ represents the set of directed edges.
The weighted adjacency matrix $\mathcal{W} = (w_{ij}) \in \mathbb{R}^{N \times N}$ satisfies $w_{ii} = 0$, and $w_{ij} > 0$ if $(v_{j}, v_{i}) \in \mathcal{E}$, indicating that agent $i$ can receive information from agent $j$.
The neighbor set of agent $i$ is defined as $\mathcal{N}_{i} = \left\{ v_{j} \in \mathcal{V} \mid (v_{j}, v_{i}) \in \mathcal{E} \right\}$.

Consider a discrete-time nonlinear MAS consisting of $N$ agents with dynamics described by
\begin{equation}\label{MAS}
x_{i}(t+1) = f(x_{i}(t), u_{i}(t)), \quad t \in \mathbb{N},
\end{equation}
where $i \in \mathcal{N}:=\{i | i =1,2,\ldots, N\}$, $x_{i}(\cdot) \in \mathcal{X} \to \mathbb{R}^n$ denotes the state of agent $i$, and $u_{i}(\cdot) \in \mathcal{U} \to \mathbb{R}^q$ represents its control input, subject to constraints: 
\begin{align}
\mathcal{X}:&=\{x_{i} \in \mathbb{R}^{n}: x_{\mathrm{min}} \leq x_{i} \leq x_{\mathrm{max}}\}, \\
\mathcal{U}:&=\{u_{i} \in \mathbb{R}^{q}: u_{\mathrm{min}} \leq u_{i} \leq u_{\mathrm{max}}\}, \label{U}
\end{align}
with $x_{\mathrm{min}} \in \mathbb{R}^{n}, x_{\mathrm{max}} \in \mathbb{R}^{n}$, $u_{\mathrm{min}} \in \mathbb{R}^{q}$ and $u_{\mathrm{max}} \in \mathbb{R}^{q}$. $f(\cdot): \mathbb{R}^n \times \mathbb{R}^q \to \mathbb{R}^n$ is locally Lipschitz continuous.

To describe the relative degree of system \eqref{MAS}, we consider an output function $z_{i}(t) = g(x_{i}(t))$.
\begin{definition}[\normalfont \textit{Relative degree} \cite{1166537}] \label{def1}
The relative degree of the output $z_{i}(t) = g(x_{i}(t))$ of system \eqref{MAS} is said to be $m$ in relation to the control input $u_{i}(t)$ if for $\forall t \in \mathbb{N}$, there holds
\begin{equation*}
\begin{aligned}
& \frac{\partial}{\partial u_{i}(t)}[g \circ \bar{f}^{l}(f(x_{i}(t),u_{i}(t)))] = \mathbf{0}_{q}, 0 \leq l \leq m-2, \\
& \frac{\partial}{\partial u_{i}(t)}[g \circ \bar{f}^{m-1}(f(x_{i}(t),u_{i}(t)))] \neq \mathbf{0}_{q}, 
\end{aligned}
\end{equation*}
where $g \circ \bar{f}^{l}(\cdot) = g(\bar{f}^{l}(\cdot))$ with $\bar{f}^{0}(x_{i}(t)) = x_{i}(t)$, $\bar{f}(x_{i}(t))$ denotes the uncontrolled state dynamics $f(x_{i}(t),0)$ and $\bar{f}^{l}$ is the $l$-times recursive compositions of $\bar{f}$ in the sense that $\bar{f}^{l}(x_{i}(t)) = f(\bar{f}^{l-1}(x_{i}(t)))$.
\end{definition}

Definition \ref{def1} implies that the control input $u_{i}(t)$ does not directly affect $z_{i}(t)$ but instead influences it after $m$ time steps.
To describe the safety of agent $i$ under system \eqref{MAS} with a relative degree $m$, we define several safe sets and the DHCBF as follows.

Agent $i$ is deemed safe if it initially belongs to the set $\mathcal{C}$ and stays within $\mathcal{C}$ at all future time steps, where $\mathcal{C}$ is defined as:
$$\mathcal{C}:=\left\{x_{i} \in \mathbb{R}^{n} \mid h(x_{i})\geq0 \right\},$$
where $h: \mathbb{R}^n \to \mathbb{R}$ is a continuously differentiable function.

If the relative degree of $h(x_{i}(t))$ to system \eqref{MAS} is $m$, we recursively define the functions $\psi_{l}: \mathbb{R}^{n} \to \mathbb{R}$ for $l = 1,\ldots,m$ as follows:
\begin{equation} \label{psi}
\psi_{l}(x_{i}(t)):=\Delta \psi_{l-1}(x_{i}(t),u_{i}(t)) + \alpha_{l}(\psi_{l-1}(x_{i}(t))),
\end{equation}
where $\psi_{0}(x_{i}(t)):= h(x_{i}(t))$, $\Delta \psi_{l-1}(x_{i}(t),u_{i}(t)) := \psi_{l-1}(x_{i}(t+1)) - \psi_{l-1}(x_{i}(t))$, and $\alpha_{l}(\cdot)$ is a class $\kappa$ function, which is chosen as $\alpha_{l}(\psi_{l-1}(x_{i}(t))) = \phi \psi_{l-1}(x_{i}(t))$ with $\phi \in(0,1] $. The corresponding safe sets are defined as:
\begin{equation}\label{cl}
\mathcal{C}^{l}:=\{x_{i} \in \mathbb{R}^{n} \mid \psi_{l}(x_{i}) \geq 0\}, l = 0,\ldots,m-1.
\end{equation}

\begin{definition}[\normalfont \textit{DHCBF} \cite{9777251}] \label{def2} 
Given $\psi_{l}(x_{i}(t))$ for $l = 1,\ldots$, $m$ and the sets $\mathcal{C}^{l}$ for $l = 0,\ldots,m-1$ as defined in \eqref{psi} and \eqref{cl}, respectively, if the relative degree of a continuously differentiable function $h: \mathbb{R}^{n} \to \mathbb{R}$ to system  \eqref{MAS} is $m$ and there exists $\psi_{m}(\cdot)$ such that:
\begin{equation}\label{DHCBF}
\psi_{m}(x_{i}(t)) \geq 0, \forall x_{i}(t) \in \bigcap \limits_{l=0}^{m-1} \mathcal{C}^{l}, \forall t \in \mathbb{N},
\end{equation}
\end{definition}
then, $h$ is a DHCBF.

Then, the formation control and obstacle avoidance problem for system \eqref{MAS} with a relative degree $m$ can be defined as:
\begin{problem}[\normalfont \textit{Formation Control and Obstacle Avoidance}]\label{prob1}
MAS \eqref{MAS} is said to achieve formation control and obstacle avoidance if for $\forall \epsilon \geq 0$, there exists a controller $u_{i}(t)$ and a time constant $T>0$ such that the following conditions hold for $\forall i \in \mathcal{N}$ and $j \in \mathcal{N}_{i}$:
\begin{align}
& \|x_{ij}(t)-d_{ij}\| \leq \epsilon, t \geq T, \label{formation} \\ 
& x_{i}(t) \in \mathcal{C}^{l}, l = 0, \ldots, m-1, t \geq 0,
\end{align}
where $x_{ij} = x_{i}-x_{j}$ and $d_{ij} \in \mathbb{R}^{n}$ is the desired relative state between agent $i$ and agent $j$. If necessary, we define $y_{ij} = x_{ij}-d_{ij}$.
\end{problem}

To achieve stable formation control, the DCLF is introduced as follows:

\begin{definition}[\normalfont \textit{DCLF}] \label{def3}
A continuously differentiable function $v: \mathbb{R}^{n} \to \mathbb{R}$ is a DCLF for MAS \eqref{MAS} if there exist contants $c_{1} >0$, $c_{2}>0$ and $0 < \lambda' \leq 1$ such that for $\forall y_{ij} \in \mathbb{R}^{n}$, $c_{1}\|y_{ij}\|^{2} \leq v(y_{ij}) \leq c_{2}\|y_{ij}\|^{2}$ and
\begin{equation}\label{DCLF}
\Delta v(y_{ij}(t)) + \lambda' v(y_{ij}(t)) \leq 0,
\end{equation}
where $i \in \mathcal{N}$, $j \in \mathcal{N}_{i}$ and $\Delta v(y_{ij}(t)) := v(y_{ij}(t+1)-v(y_{ij}(t)))$.
\end{definition}

\section{Distributed Safety-critical Model Predictive Control}
\label{3}
In this section, a DSMPC algorithm based on DHCBFs and DCLFs is proposed to achieve formation control and obstacle avoidance, which can enable distributed operation by introducing the estimated neighbor states and compatibility constraints.

At each time step $t$, Problem \ref{prob1} is formulated as the following optimization problem over a horizon $T_{p}$.

\noindent\rule{\columnwidth}{0.4pt}
\textbf{Centralized MPC:}
\begin{subequations}\label{optimization}
\begin{align}
 \min_{x_{i},u_{i}} & \sum_{i=1}^{N} \big( \sum_{k=0}^{T_{p}-1} J_{i}(x_{i}(k | t),u_{i}(k | t)) + L_{i}(x_{i}(T_{p}|t)) \big) \notag \\
s.t. \quad & x_{i}(k+1|t) = f(x_{i}(k|t),u_{i}(k|t)), \\
\quad & u_{i}(k|t) \in \mathcal{U}, x_{i}(k|t) \in \mathcal{X}, \\
\quad & x_{i}(k|t) \in \mathcal{C}^{l}, l = 0, \ldots, m-1, \label{safetycst} \\
\quad & \|y_{ij}(T_{p}|t)\| \in \mathcal{T}. \label{terminal}
\end{align}
\end{subequations}
\noindent\rule{\columnwidth}{0.4pt}
where $i \in \mathcal{N}$, $j \in \mathcal{N}_{i}$, $T_{p}$ is the prediction horizon. $x_{i}(k|t)$ represents the predicted state at time step $t+k$, starting from the current state $x_{i}(t)$, similarly for $u_{i}(k|t)$. $J_{i}(\cdot)$ and $L_{i}(\cdot)$ are the stage and terminal cost function, respectively. \eqref{safetycst} and \eqref{terminal} are safety and terminal constraints, where $\mathcal{T}$ is the terminal set. Their explicit expressions will be given by \eqref{dhcbfa} and \eqref{tv}.

Since states of agent $i$ and $j$ are coupled in \eqref{terminal}, and agent $i$ does not have access to $x_{j}(T_{p}|t)$ at the current time step $t$, to optimize in a distributed manner, we need agent $j$ to estimate its future states and transmit them to agent $i$ at time step $t$.
To address this, and considering that each agent simultaneously acts as a neighbor to others, we introduce estimated states and control inputs for each agent, which their neighbors will use during prediction and optimization.
\begin{enumerate}
\item $x_{i}^{a}(k|t)$: estimated state, $k = 1,2,\ldots, T_{p}$;
\item $u_{i}^{a}(k|t)$: estimated control input, $k = 0,1,\ldots, T_{p}-1$;
\end{enumerate}
where $i \in \mathcal{N}$.

At time step $t$, the estimated control inputs $u_{i}^{a}(k|t)$ are generated by updating the control sequence from the prior time step. Specifically, the optimal control input sequence of agent $i$ is derived by solving the optimization problem (Note that the optimization problem here refers to \eqref{DSMPC}, rather than \eqref{optimization}), and is defined as:
\begin{equation}\label{optimalu}
\mathbf{u}^{\ast}_{i}=\{u^{\ast}_{i}(0|t),u^{\ast}_{i}(1|t),\dots,u^{\ast}_{i}(T_{p}-1|t)\}.
\end{equation}
Then, the first input value $u_{i}^{*}(0|t)$ is discarded, and a new control input, computed via a predefined function $\tau(\cdot)$, is appended to the end of the sequence:
\begin{equation}
\begin{aligned}\label{estimatedu_revised}
u_{i}^{a}&(k|t+1)  \\
& =\left\{\begin{array}{ll}
	u_{i}^{*}(k+1|t),  k=0,1,\ldots,T_{p}-2,      \\
	\tau (x_{i}^{*}(T_{p}|t),\{x_{j}^{*}(T_{p}|t)\}_{j \in \mathcal{N}_{i}}),  k=T_{p}-1,
\end{array}\right.
\end{aligned}
\end{equation}
where $x_{i}^{\ast}(T_{p}|t)$ is the optimal value of $x_{i}(T_{p}|t)$, $\tau(\cdot)$ computes the final estimated control input at $t+1$ based on optimal states at the prediction time step $T_{p}$ of time $t$. Specifically, $\tau(\cdot)$ is defined as:
\begin{equation} \label{unp_revised}
\begin{aligned}
    \tau  (x_{i}^{*}(T_{p}|t),\{x_{j}^{*}&(T_{p}|t)\}_{j \in \mathcal{N}_{i}}) \\
    & =\tilde{u}_{i}(x_{i}^{*}(T_{p}|t),\{x_{j}^{*}(T_{p}|t)\}_{j \in \mathcal{N}_{i}}),
\end{aligned}
\end{equation}
with $\tilde{u}_{i}$ being a consensus-based formation control protocol:
\begin{equation} \label{nominalu_revised}
\tilde{u}_{i} (x_{i}, \{x_{j}\}_{j \in \mathcal{N}_{i}})= -K \sum_{j \in \mathcal{N}_{i}}w_{ij}y_{ij}.
\end{equation}
where $K \in \mathbb{R}^{q \times n}$ denotes the control gain matrix. 
\begin{remark}
While \eqref{nominalu_revised} defines $\tilde{u}_{i}$ as a linear consensus-based controller, alternative consensus-based protocols, such as the distributed nonlinear control protocol in \cite{10497839}, can also be employed. It essentially serves as a nominal controller and does not affect the final computed control input. Additionally, while we focus on the coupling constraints in \eqref{terminal}, collision avoidance between agents may also introduce coupling in \eqref{safetycst}. In such cases, a similar control approach can be applied to decouple \eqref{safetycst}, enabling the DHCBF to be formulated in a completely distributed way.
\end{remark}

The estimated states are recursively updated as follows:
\begin{equation} \label{estimatedstate_revised}
x_{i}^{a}(k+1|t) = f(x_{i}^{a}(k|t), u_{i}^{a}(k|t)),
\end{equation}
with the initial estimated state sequence set as $x_{i}^{a}(k|0) = x_{i}(0)$, $k = 0,\ldots, T_{p}-1$.
 
Since each agent relies on estimated states for decision-making, errors between real and estimated states may accumulate. To address this, a compatibility constraint is introduced to limit the deviation:
\begin{equation} \label{compatibility}
\eta_{i}(t) = \frac{\gamma \sum_{j \in \mathcal{N}_{i}} \|y_{i j}^{\ast}(t)\|_{Q}}{(\delta T_{p}-\delta)\sum_{j \in \mathcal{N}_{i}} \zeta_{ij}(t)},
\end{equation}
where $\gamma \in [0,1)$, $\delta$ denotes the sampling time, $y_{ij}^{\ast}(t) = x_{i}^{\ast}(t)-x_{j}^{\ast}(t)-d_{ij}$ represents the optimal value of $y_{ij}(t)$, and $\zeta_{ij}(t)$ is defined as
\begin{equation} \label{zeta}
\zeta_{ij}(t) = \max_{k \in [1,T_{p}-1]} \| x_{j}^{a}(k|t)- x_{i}^{a}(k|t)-d_{ji} \|. 
\end{equation}
Constraint \eqref{compatibility} will be used to ensure that the estimated states do not diverge excessively from the real state.

Based on the definition of the estimated states, we now proceed to present our DSMPC algorithm. Specifically, the safety constraint \eqref{safetycst} is constructed using a DHCBF. The terminal constraint \eqref{terminal} is formulated in a distributed manner by incorporating the estimated states of neighboring agents through a DCLF.

For safety, each agent $i$ has a sequence of safe sets $\mathcal{C}^{l}$ for $l = 0, \ldots, m-1$, as defined in \eqref{cl}. If we define $h(x_{i}): \mathbb{R}^{n} \to \mathbb{R}$ as a DHCBF with relative degree $m$ for agent $i$, then any Lipschitz continuous controller $u_{i}$ satisfying the following condition will ensures that $\mathcal{C}^{0} \cap \cdots \cap{C}^{m-1}$ remains forward invariant for agent $i$ with dynamics \eqref{MAS}:
\begin{equation}\label{dhcbfa}
\psi_{m}(x_{i}(t),u_{i}(t)) \geq 0, t \geq 0.
\end{equation}
This guarantees that states of agent $i$ satisfy $x_{i}(t) \in \mathcal{C}^{0} \cap \cdots \cap{C}^{m-1}, \forall t \geq 0$. Thus, we will replace \eqref{safetycst} with \eqref{dhcbfa} and treat it as a safety constraint.

For formation control, we define a DCLF to construct the terminal constraint as follows:
\begin{align} 
v(\tilde{y}_{ij}) &= \sum_{j \in \mathcal{N}_{i}}\|\tilde y_{ij}\|^{2}, \label{vv} \\
v(\tilde{y}_{ij}(T_{p}|t)) & \leq \lambda v(\tilde{y}_{ij}(T_{p}|t-1)), \label{tv}
\end{align}
where $i \in \mathcal{N}$, $j \in \mathcal{N}_{i}$, $\tilde{y}_{ij} = x_{i} - x^{a}_{j}-d_{ij}$, $\lambda = 1-\lambda'$. Consequently, \eqref{tv} will replace \eqref{terminal} as the terminal constraint.

Then, the DSMPC program can be formulated as the following optimization problem for each agent $i$:

\noindent\rule{\columnwidth}{0.4pt}
\textbf{DSMPC for agent $i$:}  
\begin{subequations} \label{DSMPC}
\begin{align}
& \min_{x_{i},u_{i}} J_{i}(x_{i}(k|t),u_{i}(k|t)) + L_{i}(x_{i}(T_{p}|k)) \label{DSMPCobj} \\
 \text{s.t.} \quad & x_{i}(k+1|t) = f(x_{i}(k|t),u_{i}(k|t)), \label{DSMPCdynamics} \\
 \quad \quad & u_{i}(k|t) \in \mathcal{U}, x_{i}(k|t) \in \mathcal{X}, \label{DSMPCux} \\
 \quad \quad & \psi_{m}(x_{i}(k|t),u_{i}(k|t)) \geq 0, \label{DSMPCdhcbf}\\
 \quad \quad & \| \tilde{x}_{ii}(k|t) \|_{Q} \leq \eta_{i}(t), k \in [1,T_{p}-1] \label{DSMPCcom}\\
  \quad \quad & v(\tilde{y}_{ij}(T_{p}|t)) \leq 
\begin{cases}
  \lambda^{T_{p}} v(\tilde{y}_{ij}(T_{p}|0)),  & t = 0\\
  \lambda v(\tilde{y}_{ij}(T_{p}|t-1)), & t > 0
\end{cases} \label{DSMPCv}
\end{align}
\end{subequations}
\noindent\rule{\columnwidth}{0.4pt}
where \eqref{DSMPCobj} is the cost function of agent $i$:
\begin{align}
J_{i} &=\sum_{k=0}^{T_{p}-1}(\sum_{j \in \mathcal{N}_{i}}\|\tilde{y}_{ij}(k|t)\|_{Q} + \|u_{i}(k|t)\|_{R}), \label{stagecost} \\
L_{i} &= \sum_{j \in \mathcal{N}_{i}} \|\tilde{y}_{ij}(T_{p}|t)\|_{Q},
\end{align}
with $Q\succ 0$, $R \succ 0$. The constraint \eqref{DSMPCdhcbf} ensures the safety of agent $i$ over the prediction horizon $T_{p}$. The compatibility constraint \eqref{DSMPCcom} limits the estimated error. The terminal constraint $\eqref{DSMPCv}$ guarantees that the state deviations $\tilde{y}_{ij}$ are driven toward convergence by the end of the prediction horizon.

\begin{remark} \label{rmk1}
\eqref{DSMPCux}, \eqref{DSMPCdhcbf}, and \eqref{DSMPCv} may be mutually incompatible, potentially rendering the infeasibility of the DSMPC problem \eqref{DSMPC}. To address this issue, we can introduce a non-negative penalty term $\rho \geq 0$ to enhance feasibility. Specifically, $\rho$ is incorporated into the terminal constraint \eqref{DSMPCv} as:
\begin{equation} \label{vrho}
v(\tilde{y}_{ij}(T_{p}|t)) \leq 
\begin{cases}
  \lambda^{T_{p}} v(\tilde{y}_{ij}(T_{p}|0)) + \rho,  & t = 0 \\
  \lambda v(\tilde{y}_{ij}(T_{p}|t-1)) + \rho, & t > 0
\end{cases}
\end{equation}
Accordingly, the terminal cost function is modified to:
\begin{equation} \label{terminalcost}
 L_{i}(x_{i}(T_{p}|t), \rho) = \sum_{j \in \mathcal{N}_{i}} \|\tilde{y}_{ij}(T_{p}|t)\|_{Q} + \rho^{2}.
\end{equation}
However, this relaxation only mitigates potential conflicts between \eqref{DSMPCv} and the other two constraints, \eqref{DSMPCcom} and \eqref{DSMPCdhcbf}, but does not resolve the conflict between \eqref{DSMPCcom} and \eqref{DSMPCdhcbf}. A detailed feasibility analysis in Theorem \ref{them1} will further clarify that \eqref{DSMPCcom} can always be satisfied without violating \eqref{DSMPCdhcbf}.\end{remark}

At time step $t$, the optimal control input sequence of \eqref{DSMPC} is given by \eqref{optimalu}. By applying the first input $u_{i}^{\ast}(0|t)$ to system \eqref{MAS}, the subsequent state is obtained:
\begin{equation}\label{transstate}
x_{i}^{\ast}(t+1)=x^{\ast}_{i}(1|t)=f(x_{i}^{\ast}(0|t),u_{i}^{\ast}(0|t)),
\end{equation}
where $x_{i}^{\ast}(t+1)$ is the optimal state of agent $i$ at time $t+1$. This process is repeated in a receding horizon fashion: at each time step, \eqref{DSMPC} is solved, the first control input is applied, and the optimization is updated based on the new state. The detailed implementation of this iterative procedure is outlined in Algorithm \ref{alg1}.

\begin{algorithm}[h] 
\caption{DSMPC algorithm} 
\label{alg1} 
\begin{algorithmic}[1]
\REQUIRE $\mathcal{U}$; $\mathcal{X}$; $T_{p}$; $h$; $v$; $d_{ij}$, $\epsilon$.
\ENSURE $\mathbf{u}_{i}^{\ast}$.
\STATE Set $t=0$;
\STATE Initialize real and estimated states of all agents;
\WHILE{$\|y_{i j}\| > \epsilon$}
\STATE Receive $x_{j}^{a}(k|t)$, $k = 0:T_{p}$, $j \in \mathcal{N}_{i}$;
\STATE Solve \eqref{DSMPC} at time step $t$ to obtain $\mathbf{u}_{i}^{\ast}$;
\STATE Substitute $u_{i}^{\ast}(0|t)$ to \eqref{MAS};
\STATE Update agent state $x_{i}^{\ast}(t+1)$ by \eqref{transstate};
\STATE Compute $x_{i}^{a}(k|t+1)$ using \eqref{estimatedstate_revised};
\STATE Share $x_{i}^{a}(k|t+1)$ to neighbors;
\STATE $t \leftarrow t+1$;
\ENDWHILE
\end{algorithmic}
\end{algorithm}

Under the assumption on the bound of $\psi_{m-1}(x_{i}(t))$, we prove the feasibility of the DSMPC program \eqref{DSMPC} and the stability of Algorithm \ref{alg1}.
\begin{assumption}\label{asp1}
Given $\psi_{l}(x_{i}), l = 1,\ldots$, $m$, and $\mathcal{C}^{l}$, $l = 0,\ldots,m-1$, defined by \eqref{psi} and \eqref{cl}, respectively. For a given $f(\cdot): \mathbb{R}^{n} \to \mathbb{R}^{n}$ and a DHCBF candidate $h(\cdot): \mathbb{R}^{n} \to \mathbb{R}$ with relative degree $m$, we assume that:
\begin{equation} \label{psim}
\begin{aligned}
& \psi_{m-1} \circ f(x_{i}, u_{i}^{M}) = \sup_{u_{i} \in \mathcal{U}}[\psi_{m-1} \circ f(x_{i},u_{i})], \\
& \psi_{m-1} \circ f(x_{i},u_{i}^{M}) - \psi_{m-1}(x_{i}) + \alpha_{m}(\psi_{m-1}(x_{i})) \geq 0,
\end{aligned}
\end{equation}
\end{assumption}
where $u_{\mathrm{min}} \leq u_{i}^{M} \leq u_{\mathrm{max}}$.
\begin{figure*}[t]
\centering
\subfigure[]{
		\includegraphics[scale=0.23]{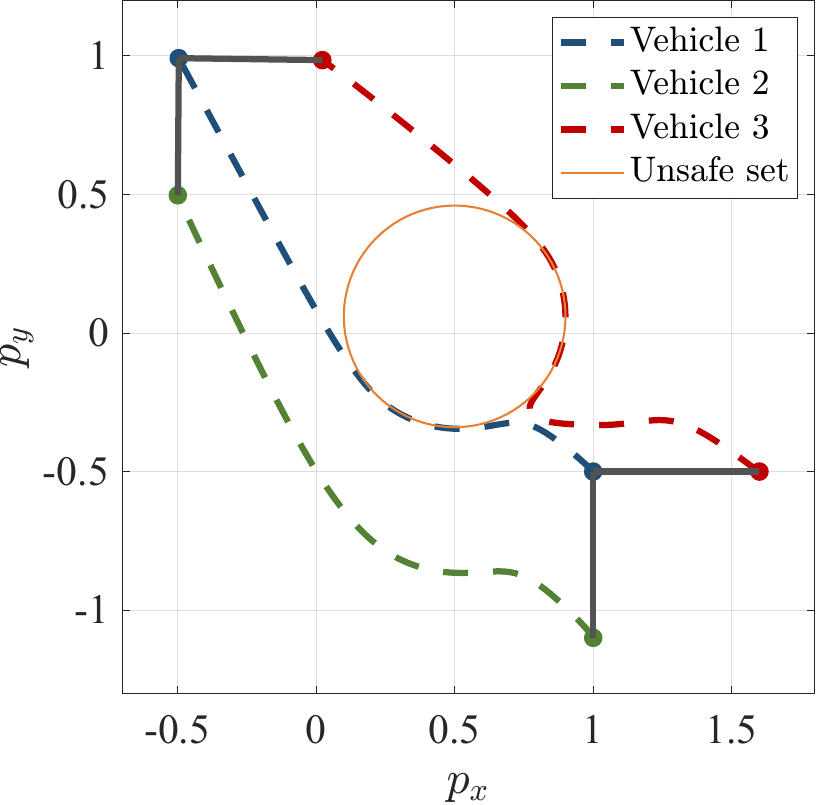}}
\subfigure[]{
		\includegraphics[scale=0.23]{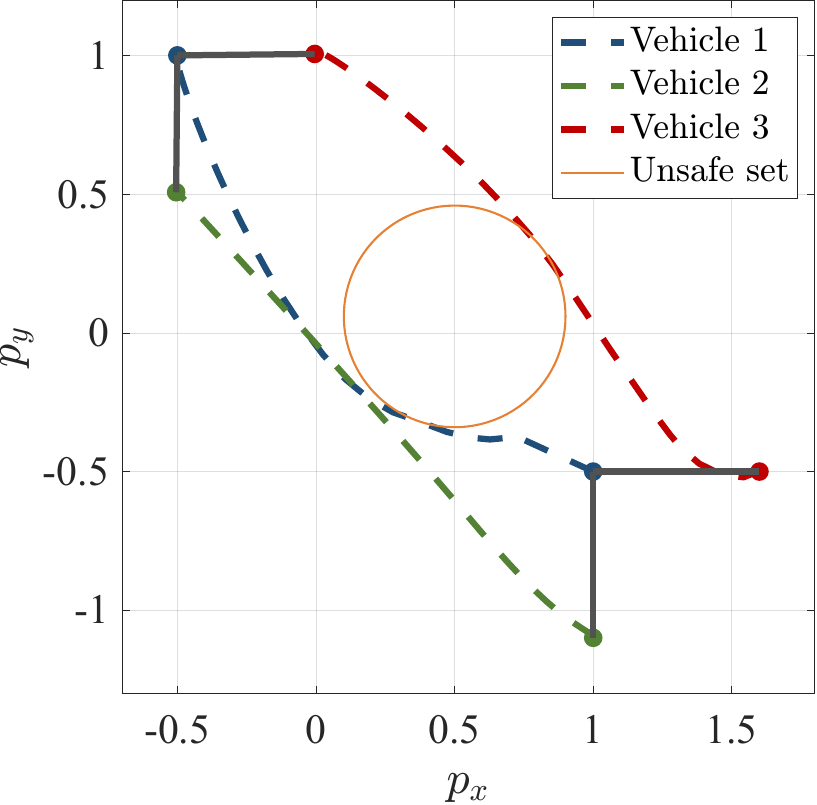}}
\subfigure[]{
		\includegraphics[scale=0.23]{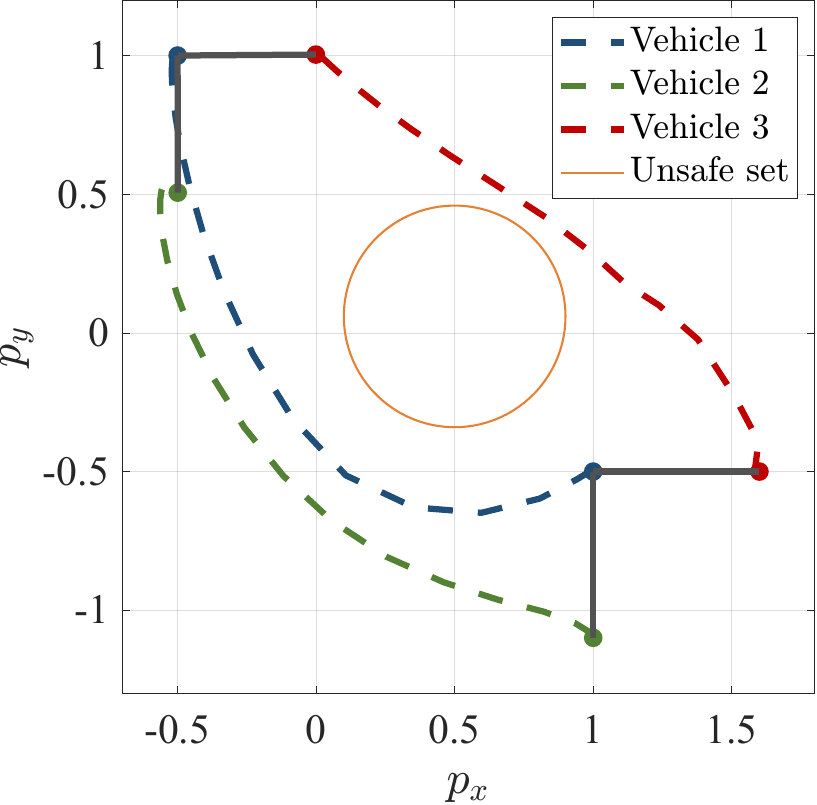}}
\subfigure[]{
        \includegraphics[scale=0.23]{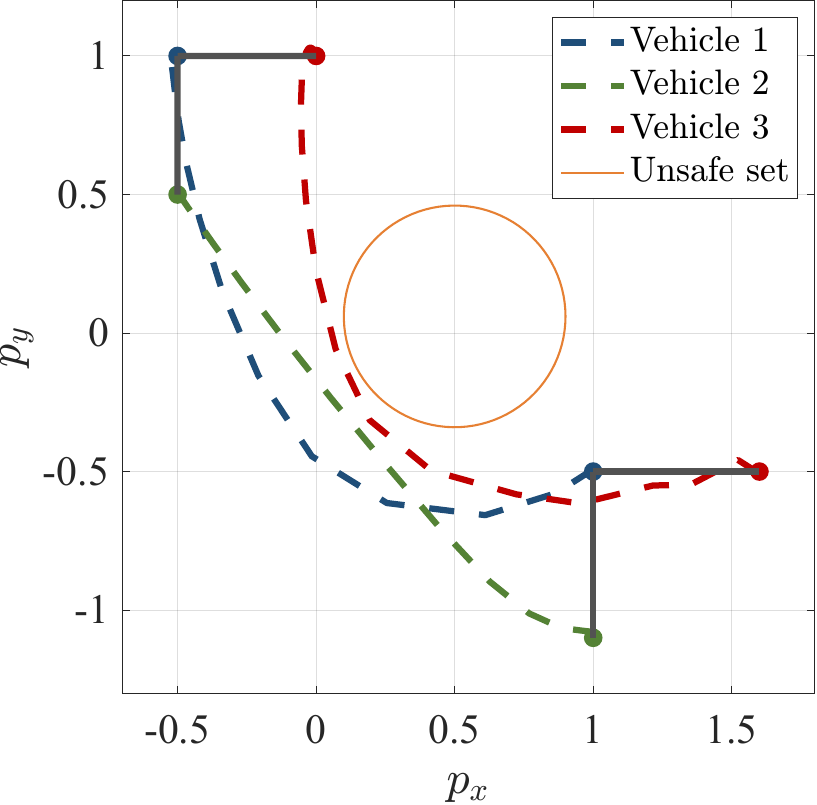}}
\subfigure[]{
        \includegraphics[scale=0.23]{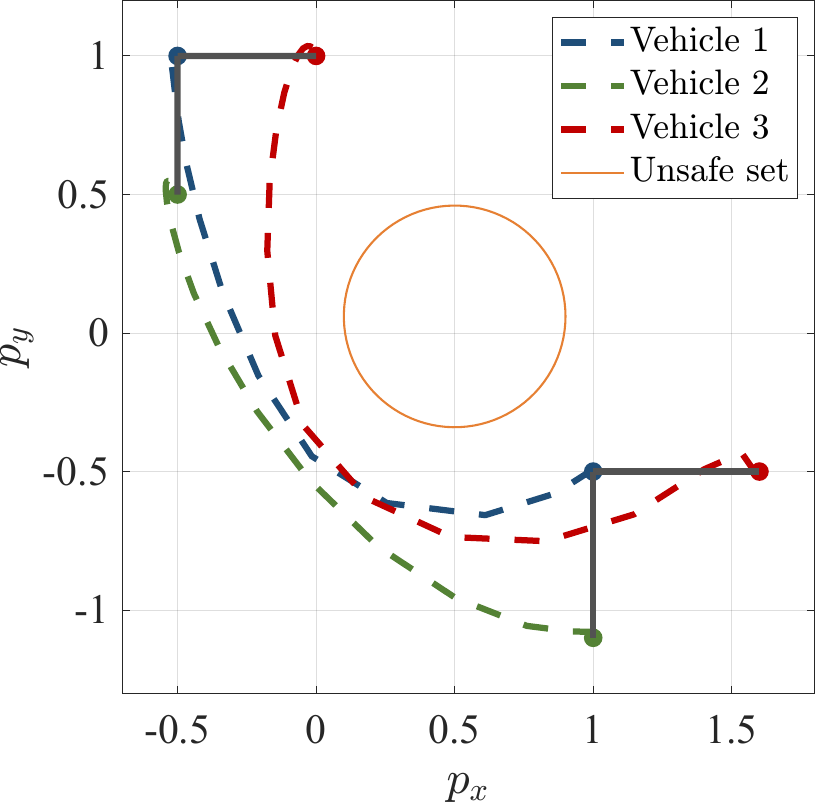}}
\caption{A multi-vehicle formation avoids an obstacle using different methods. (a) The NC-CBF method; (b) The MPC-DC method with $T_{p}=20$; (c) The DSMPC algorithm with $T_{p}=3$ and $\gamma = 0.1$; (d) The DSMPC algorithm with $T_{p}=5$ and $\gamma = 0.1$; (e) The DSMPC algorithm with $T_{p}=5$ and $\gamma = 0.8$.}
\label{fig4}
\end{figure*}
\begin{remark}
Assumption \ref{asp1} is a mild assumption because it only requires the control input to be selected from the feasible set to ensure the non-negativity of $\psi_{m-1}(x_{i})$, without imposing excessive restrictions on the system dynamics or feasible control strategies. This assumption is similar to those commonly used in classical CBF design \cite{7040372, 10145590}, which is widely applied in the safety-critical control and can typically be satisfied by appropriately choosing the parameter $\alpha_{m}(\cdot)$. Therefore, in practical applications, it does not constitute an overly restrictive constraint.
\end{remark}

Based on Assumption \ref{asp1}, we analyze feasibility and stability as follows:

\begin{theorem}\label{them1}
Suppose Assumption \ref{asp1} holds. If DSMPC program \eqref{DSMPC} is initially feasible, then it will remain feasible at all future time steps.
\end{theorem}
\begin{proof}
See the Appendix.
\end{proof}

\begin{remark}\label{rmk3}
The MPC-CBF method in \cite{9483029} enhances feasibility through an adaptive factor, which is a meaningful attempt toward practical applicability, though it lacks theoretical feasibility guarantees. 
In contrast, our proposed DSMPC framework not only preserves the advantages of distributed computation and safety-critical guarantees but also provides rigorous theoretical assurances of recursive feasibility. Furthermore, the assumption that DSMPC \eqref{DSMPC} is initially feasible is mild and commonly adopted in related works \cite{8039204, 9779571}. Therefore, despite the nonlinear and distributed nature of \eqref{DSMPC}, our approach remains computationally tractable.
\end{remark}

\begin{theorem}\label{them2}
Suppose Assumption \ref{asp1} holds. The formation control and obstacle avoidance can be achieved via controllers computed by Algorithm \ref{alg1}.
\end{theorem}

\begin{proof}
See the Appendix.
\end{proof}

\section{Simulation}
\label{5}

This section validates the DSMPC algorithm in a multi-vehicle system for formation control and obstacle avoidance. Besides, we compare its effectiveness against other methods in \cite{2023centralized,7857061, 7040372}, highlighting that DSMPC achieves improved performance and reduced computation time.

Consider a multi-vehicle system with $N=3$ under a directed communication topology, where the edge set is given by $\mathscr{E} = \{(2,1), (3,1)\}$. The dynamics of each vehicle are described as follows:
\begin{equation}\label{ex1}
\begin{bmatrix}
    p_{x_{i}}(t+1) \\ p_{y_{i}}(t+1) \\ v_{x_{i}}(t+1) \\ v_{y_{i}}(t+1)
\end{bmatrix}
=
\begin{bmatrix}
    p_{x_{i}}(t) + \delta  v_{x_{i}}(t) \\ p_{y_{i}}(t)+\delta v_{y_{i}}(t) \\ v_{x_{i}}(t)+\delta (av^{3}_{x_{i}}(t) + u_{x_{i}}(t))\\ v_{y_{i}}(t)+\delta(av^{3}_{y_{i}}(t)+u_{y_{i}}(t))
\end{bmatrix}
\end{equation}
where $i = 1,2,3$, $a = -3$ and $\delta = 0.1$. The position of the $i$-th vehicle is $(p_{x_{i}},p_{y_{i}}) \in \mathbb{R}^{2}$, while its velocity is given by $(v_{x_{i}},v_{y_{i}}) \in \mathbb{R}^{2}$. The state is denoted by $x_{i} = [p_{x_{i}}, p_{y_{i}}, v_{x_{i}}, v_{y_{i}}]$ and the control input is $u_{i} = [u_{x_{i}}, u_{y_{i}}]$. The state and input bounds are $[\pm 5, \pm 5, \pm 2, \pm 2] $ and $[\pm 0.5, \pm 0.5]$, respectively.

The DSMPC algorithm is applied to \eqref{ex1} to achieve formation control while avoiding an obstacle. The corresponding stage and terminal cost function of the $i$-th vehicle is given by \eqref{stagecost} and \eqref{terminalcost},\ respectively, with $Q = \mathcal{I}_{4}$ and $R = \mathcal{I}_{2}$. The DHCBF and the DCLF are chosen as follows:
\begin{align*}
& h(x_{i}) = (x_{i}-x_{m})^{2} -r^{2}_{m}, \\
& v(x_{i}) = (x_{i}-x^{a}_{j}-d_{ij})^{2},
\end{align*}
where $d_{12} = [0,0.5]$, $d_{13} = [-0.5,0]$, $x_{m}$ and $r_{m}$ denote the position and radius of the obstacle. The parameters of the DHCBF are chosen as $\phi_{1} = 0.1$, $\phi_{2} = 0.9$, $\phi_{3} = 0.4$, where $\phi_{i}$ represents the parameter in \eqref{DSMPCdhcbf} corresponding to the $i$-th vehicle. Additionlly, the parameters $\gamma$ in \eqref{DSMPCcom} and $\lambda$ in \eqref{DSMPCv} are chosen as $0.1$  and $0.9$, respectively. A classic consensus-based control protocol \eqref{nominalu_revised} is employed to compute the estimated states. 
The initial positions, target positions, and the location of the obstacle are illustrated in Fig. \ref{fig4}. 
\begin{table}[t]
    \centering
    \small 
    \renewcommand{\arraystretch}{1} 
    \setlength{\tabcolsep}{3pt} 
    \caption{Different methods benchmark in terms of feasibility, average computation time (act), minimum distance from the obstacle  (min d), maximum communication range (max r), and cost.}  
    \label{tab1} 
    \begin{tabular}{l l c c c c c c}
        \toprule
        controller & status & $T_{p}$ & $\gamma$ & act (s) & min d& max r& cost \\
        \midrule
        \multirow{5}{*}{DSMPC}   & solved & 2 & 0.1 & \textbf{0.289$\pm$0.002} & 0.053 & 1.26 & 177.5 \\
                                 & solved & 3 & 0.1 & 0.355$\pm$0.001 & 0.085 & 1.29 & \textbf{144.3} \\
                                 & solved & 5 & 0.1 & 0.488$\pm$0.002 & 0.046 & 1.29 & 192.9 \\
                                 & solved & 5 & 0.4 & 0.577$\pm$0.002 & 0.245 & 1.31 & 182.9 \\
                                 & solved & 5 & 0.8 & 0.482$\pm$0.003 & \textbf{0.251} & \textbf{1.21} & 171.7 \\
        \midrule
        \multirow{3}{*}{MPC-DC}  & infeas. & 15  & -   & NaN  & NaN  & NaN  & NaN \\
                                 & solved  & 20  & -   & 1.290$\pm$0.011 & 0.000 & 1.362 & 536.3\\
                                 & solved  & 30 & -   & 1.448$\pm$0.008 & 0.000 & 1.873 & 969.2 \\
        \midrule
        \multirow{1}{*}{NC-CBF}  & solved & -  & - & 0.385$\pm$0.002 & 0.000 & 1.323 & 416.7 \\
        \midrule
        \multirow{1}{*}{CLF-CBF} & infeas. & -  & - & NaN  & NaN  & NaN  & NaN \\
        \bottomrule
    \end{tabular}
\end{table}

We compare the DSMPC algorithm with other methods: (1) The MPC-DC from \cite{2023centralized}, which is also a DMPC framework but enforces safety constraints using the Eucidean norm distance; (2) The nominal controller with DHCBF (NC-DHCBF) from \cite{7857061}, which utilizes DHCBFs to enforce safety constraints and formulates the deviation between the real and nominal controllers as the objective function of a QP; (3) The CLF-CBF from \cite{7040372}.

Fig. \ref{fig4} presents the simulation results comparing all methods. In Fig. \ref{fig4} (a), the NC-CBF method exhibits abrupt maneuvers and sharp turns, indicating limited trajectory smoothness. In contrast, Fig. \ref{fig4} (b) shows that the MPC-DC method significantly enhances smoothness. However, due to the lack of proactive obstacle avoidance constraints, the minimum distance for avoiding the obstacle remains small, making the approach less safe. Fig. \ref{fig4} (c) presents the effectiveness of the DSMPC algorithm with $T_{p}=3$ and $\gamma=0.1$, where all three vehicles maintain relatively large and safe obstacle avoidance distances. However, vehicle $1$ is positioned too far from vehicle $1$. Fig. \ref{fig4} (d) improves upon this by increasing $T_{p}$ to $5$, achieving a better balance between stability and safety, resulting in tight formation maintenance. Fig. \ref{fig4} (e) demonstrates the effect of increasing $\gamma$ to $0.8$, which allows vehicles $2$ and $3$ more flexibility in maneuvering. As a result, the formation is well-maintained, and the obstacle avoidance is further improved.

In Table \ref{tab1}, we benchmark the performance of different methods (DSMPC, MPC-DC, NC-CBF, and CLF-CBF). A larger minimum distance from the obstacle indicates improved vehicle safety, while a smaller maximum communication range is preferable, as it signifies that vehicles can maintain a tighter formation. The proposed DSMPC algorithm demonstrates significant advantages by successfully solving all test cases. While the NC-CBF method is also able to solve them, it incurs a considerably higher cost. In contrast, MPC-DC (for $T_{p} \leq 15$) and CLF-CBF fail in certain scenarios. This is primarily because, in MPC-DC, The relative degree of the constraint $h \geq 0$ to the input is $2$, requiring a longer prediction horizon for effective obstacle avoidance. In the case of CLF-CBF, its centralized framework leads to coupling between $u_{i}$ and $u_{j}$ in the CLF with a relative degree $1$, making it challenging to find a feasible solution. This highlights the advantage of our DSMPC algorithm, which succeeds in all test cases thanks to its distributed nature and its ability to handle systems with high relative degrees. Additionally, the DSMPC algorithm achieves significantly lower computation times compared to MPC-DC. With appropriate parameter tuning, the DSMPC algorithm also demonstrates lower control costs than other methods. Increasing $T_{p}$ from $2$ to $5$ strikes a balance between maintaining tight formation and ensuring effective obstacle avoidance. Furthermore, a higher $\gamma$ initially reduces the control cost. Therefore, in practice, we can select $T_{p}$ and $\gamma$ based on specific requirements.

\section{Conclusion}
\label{6}
This paper has presented a DSMPC algorithm that ensures formation control and obstacle avoidance in nonlinear MASs by integrating DHCBFs and DCLFs. The approach maintains a fully distributed structure and guarantees recursive feasibility and stability. Simulation results have demonstrated its improved efficiency over existing methods.
In the future, we will attempt to extend our method to stochastic or uncertain MASs to enhance robustness.

\appendix
\subsection{Proof of Theorem 1}
To prove Theorem \ref{them1}, it only needs to show that if, at time $t$, there exists a sequence of control inputs 
\begin{equation} \label{feaut}
\mathbf{u}_{i}^{\ast}(t) = \{u_{i}^{\ast}(0|t), u_{i}^{\ast}(1,t), \ldots, u_{i}^{\ast}(T_{p}-1,t) \}.
\end{equation}
such that \eqref{DSMPC} is feasible for $\forall i \in \mathcal{N}$ and $j \in \mathcal{N}_{i}$, then there exists a sequence of control inputs $\mathbf{u}_{i}(t+1)$ ensuring that it remains feasible at time $t+1$. 

Assume that at time $t$, for $\forall i \in \mathcal{N}$ and $j \in \mathcal{N}_{i}$, \eqref{DSMPC} is feasible. Then, the optimal control input sequence at time $t$ can be represented as $\eqref{feaut}$. It follows from $\eqref{transstate}$ that
\begin{align}
x_{i}^{\ast}(k|t+1) = x_{i}^{\ast}(k+1|t),
\end{align}
where $0 \leq k \leq T_{p}-1$, and
\begin{align}
u_{i}^{\ast}(k|t+1) = u_{i}^{\ast}(k+1|t),
\end{align}
where $0 \leq k \leq T_{p}-2$. If we can find a control input $\bar{u}_{i}(T_{p}-1|t+1)$ such that \eqref{DSMPC} remains feasible at time step $T_{p}-1$ predicted at time step $t+1$, then we can represent
\begin{align} \label{feaut+1}
\mathbf{u}_{i} & (t+1) =  \\
 & \{u_{i}^{\ast}(1|t), u_{i}^{\ast}(2,t), \ldots, u_{i}^{\ast}(T_{p}-1,t), \bar{u}_{i}(T_{p}-1|t+1) \}, \notag 
\end{align}
as a feasible control input sequence of \eqref{DSMPC} at time $t+1$.
Based on \eqref{psim} in Assumption \ref{asp1}, there exists at least one control input $\bar{u}_{i}(T_{p}-1|t+1) = u_{i}^{M} \in \mathcal{U}$ such that
\begin{equation} \label{psimm}
\begin{aligned}
\psi_{m}(x_{i}(T_{p}|t+1)) & = \psi_{m-1}(f(x_{i}(T_{p}-1|t+1),u_{i}^{M})) \\
& - \psi_{m-1}(x_{i}(T_{p}-1|t+1)) \\
& +\alpha(\psi_{m-1}(x_{i}(T_{p}-1|t+1))) \geq 0,
\end{aligned}
\end{equation} 
implying that the constraint \eqref{DSMPCdhcbf} is satisfied at the predictive time $T_{p}-1$ and does not conflict with the constraint \eqref{DSMPCux}. 
Next, we prove that \eqref{DSMPCcom} is satisfied at time $t+1$. From \eqref{feaut+1}, we have
\begin{equation} \label{xtrans}
\begin{aligned}
x_{i}(k|t+1) & = f(x_{i}^{\ast}(k-1,t+1),u_{i}(k-1,t+1)) \\
& = f(x_{i}^{\ast}(k,t),u_{i}^{\ast}(k,t)) \\  
& = x_{i}^{\ast}(k+1|t),
\end{aligned}
\end{equation}
where $k = 0,1,\ldots,T_{p}-1$. By \eqref{estimatedu_revised}, \eqref{transstate} and \eqref{feaut+1} with $x_{i}^{a}(0|t+1)=x_{i}^{*}(1|t)$, there holds:
\begin{equation}\label{xa*}
x_{i}^{a}(k|t+1) = x_{i}^{\ast}(k+1|t),
\end{equation}
where $k = 0,1,\ldots,T_{p}-1$. Thus, based on \eqref{xtrans} and \eqref{xa*}, it can be found that the left side of constraints \eqref{DSMPCcom} equals zero at time step $t+1$, i.e.,
\begin{equation}
\begin{aligned}
\|\tilde{x}_{ii}(k|t+1)\|_{Q} & = \|x_{i}(k|t+1) - x_{i}^{a}(k|t+1)\|_{Q} \\
                                      & = 0.
\end{aligned}
\end{equation}
Meanwhile, as $\delta T_{p} >\delta$, according to the compatibility constraints \eqref{compatibility}, $\eta(t+1) \geq 0$ is always guaranteed. Therefore, the constraint \eqref{DSMPCcom} are satisfied at time step $t+1$, i.e.,
\begin{equation}
 \tilde{x}_{ii}(k|t+1)\|_{Q} \leq \eta(t+1).
\end{equation}
As $v$ is a DCLF, based on \eqref{DCLF}, there exists a $\rho \geq 0$ satisfying the terminal constraints \eqref{DSMPCv} if $\rho = 0$ or \eqref{vrho} if $\rho > 0$. Thus, \eqref{feaut+1} is a feasible solution to \eqref{DSMPC} at time step $t+1$. Since the optimization problem \eqref{DSMPC} is initially feasible, by the principle of recurrence, \eqref{DSMPC} will remain feasible at all future time steps.

\subsection{Proof of Theorem 2}
From \eqref{DCLF} and \eqref{DSMPCv}, for $\forall t \in \mathbb{N}$, there is 
\begin{equation} \label{tildev}
v(\tilde{y}_{ij}(T_{p}|t+1)) \leq \lambda v(\tilde{y}_{ij}(T_{p}|t)),
\end{equation}
where $i \in \mathcal{N}$ and $j \in \mathcal{N}_{i}$. By recursively applying \eqref{tildev}, we obtain:
\begin{equation}
v(\tilde{y}_{ij}(T_{p}|t+1)) \leq \lambda^{t+1} v(\tilde{y}_{ij}(0)) = \tilde{\epsilon}, t \gg 0,
\end{equation} 
where $\tilde{\epsilon}$ is a sufficiently small positive constant. Based on \eqref{vv}, there exists
\begin{equation} \label{ggt}
\|\tilde{y}_{ij}(t)\| \leq \tilde{\epsilon}, t \gg 0.
\end{equation}
Considering the compatibility constraint \eqref{DSMPCcom} and the triangle inequality property, we have
\begin{align}
\|y_{ij}(k|t)\|&= \|x_{i}(k|t)-x_{j}(k|t)-d_{ij}\| \notag \\
& = \|x_{i}^{a}(k|t)-x_{j}(k|t)-d_{ij}+x_{i}(k|t)-x_{i}^{a}(k|t)\| \notag \\
& \leq \|x_{i}^{a}(k|t)-x_{j}(k|t)-d_{ij}\|  \notag \\
& \quad + \|x_{i}(k|t)-x_{i}^{a}(k|t)\| \notag \\
& \leq \|\tilde{y}_{ij}(k|t)\| + \eta_{i}(t). \label{The4 21}
\end{align}
Then, based on \eqref{ggt}, for $t \gg 0$, there is 
\begin{equation}\label{The4 22}
\begin{aligned}
\|y_{ij}(k|t)\| & \leq \|\tilde{y}_{ij}(k|t)\| + \eta_{i}(t) \\
& \leq \tilde{\epsilon} + \eta_{i}(t) \\
& \leq \tilde{\epsilon} + \frac{\gamma \sum_{j \in \mathcal{N}_{i}} (\tilde{\epsilon} + \eta_{i}(t-1))}{(\delta T_{p}-\delta)\sum_{i \in \mathcal{N}_{i}} \zeta_{ij}(t)} \\
& = \tilde{\epsilon} + c_{1} \gamma \tilde{\epsilon} + c_{1} \gamma \eta_{i}(t-1),
\end{aligned}
\end{equation}
where $c_{1} = N_{i}/ (\delta T_{p}-\delta)\sum_{i \in \mathcal{N}_{i}} \zeta_{ij}(t) $, $N_{i}$ is the number of neighbours of agent $i$.  Based on \eqref{compatibility} and \eqref{The4 21}, it can be obtained that
\begin{equation}
\begin{aligned}
\eta_{i}(t-1) & \leq  \frac{\gamma \sum_{j \in \mathcal{N}_{i}} (\tilde{\epsilon} + \eta_{i}(t-2))}{(\delta T_{p}-\delta)\sum_{i \in \mathcal{N}_{i}} \zeta_{ij}(t-1)} \\
& =  c_{2} \gamma \tilde{\epsilon} + c_{2} \gamma  \eta_{i}(t-2),
\end{aligned}
\end{equation}
where $c_{2} = N_{i}/ (\delta T_{p}-\delta)\sum_{i \in \mathcal{N}_{i}} \zeta_{ij}(t-1)$. Thus, according to the recursion principle, we have 
\begin{equation}
\begin{aligned}
\|y_{ij}(k|t)\|& \leq  \tilde{\epsilon} (1+\sum_{q=1}^{t} (\gamma^{q} \prod_{p=1}^{q} c_{p})) + \gamma^{t} \prod_{p=1}^{t} c_{p} \eta_{i}(0),
\end{aligned}
\end{equation}
where $c_{p} = N_{i}/ (\delta T_{p}-\delta)\sum_{j \in \mathcal{N}_{i}} \zeta_{ij}(t-p+1) $, $\prod$ denotes the product of terms in tandem. Define the maximum value of $\prod_{p=1}^{q} c_{p}$ as $c_{\text{max}}$. We have
\begin{equation}
\begin{aligned}
 & \lim_{t \to \infty} (\tilde{\epsilon} (1+\sum_{q=1}^{t} (\gamma^{q} \prod_{p=1}^{q} c_{p})) + \gamma^{t} \prod_{p=1}^{t} c_{p} \eta_{i}(0)) \\
 & \leq \tilde{\epsilon}(1+\frac{\gamma c_{\text{max}}}{1- \gamma}).
 \end{aligned}
\end{equation}
As $\tilde{\epsilon}$ is small enough, yielding 
\begin{equation}
 \lim_{t\to\infty} \|y_{ij}(k|t)\| \leq \tilde{\epsilon}(1+\frac{\gamma c_{\text{max}}}{1- \gamma}) \leq \epsilon.
\end{equation}
This implies that MAS \eqref{MAS} will asymptotically achieve formation. Additionally, according to Theorem \ref{them1}, constraints \eqref{DSMPCux}, \eqref{DSMPCdhcbf}, and \eqref{DSMPCcom} will always be satisfied. Thus, the formation control and obstacle avoidance can be achieved with controllers computed by Algorithm \ref{alg1}.

\addtolength{\textheight}{-17.7cm} 

\bibliographystyle{IEEEtran}
\bibliography{ref}

\begin{thebibliography}{10}
\providecommand{\url}[1]{#1}
\csname url@samestyle\endcsname
\providecommand{\newblock}{\relax}
\providecommand{\bibinfo}[2]{#2}
\providecommand{\BIBentrySTDinterwordspacing}{\spaceskip=0pt\relax}
\providecommand{\BIBentryALTinterwordstretchfactor}{4}
\providecommand{\BIBentryALTinterwordspacing}{\spaceskip=\fontdimen2\font plus
\BIBentryALTinterwordstretchfactor\fontdimen3\font minus
  \fontdimen4\font\relax}
\providecommand{\BIBforeignlanguage}[2]{{%
\expandafter\ifx\csname l@#1\endcsname\relax
\typeout{** WARNING: IEEEtran.bst: No hyphenation pattern has been}%
\typeout{** loaded for the language `#1'. Using the pattern for}%
\typeout{** the default language instead.}%
\else
\language=\csname l@#1\endcsname
\fi
#2}}
\providecommand{\BIBdecl}{\relax}
\BIBdecl

\bibitem{8798870}
S.~G. Manyam, K.~Sundar, and D.~W. Casbeer, ``Cooperative routing for an
  air--ground vehicle team---exact algorithm, transformation method, and
  heuristics,'' \emph{IEEE Transactions on Automation Science and Engineering},
  vol.~17, no.~1, pp. 537--547, 2020.

\bibitem{9538804}
Z.~Pan, C.~Zhang, Y.~Xia, H.~Xiong, and X.~Shao, ``An improved artificial
  potential field method for path planning and formation control of the
  multi-uav systems,'' \emph{IEEE Transactions on Circuits and Systems II:
  Express Briefs}, vol.~69, no.~3, pp. 1129--1133, 2022.

\bibitem{WU2020106332}
Y.~Wu, J.~Gou, X.~Hu, and Y.~Huang, ``A new consensus theory-based method for
  formation control and obstacle avoidance of uavs,'' \emph{Aerospace Science
  and Technology}, vol. 107, p. 106332, 2020.

\bibitem{10384062}
P.~Aditya and H.~Werner, ``A distributed linear quadratic discrete-time game
  approach to formation control with collision avoidance,'' in \emph{2023 62nd
  IEEE Conference on Decision and Control (CDC)}, 2023, pp. 1239--1244.

\bibitem{2023centralized}
I.~Ravanshadi, E.~A. Boroujeni, and M.~Pourgholi, ``Centralized and distributed
  model predictive control for consensus of non-linear multi-agent systems with
  time-varying obstacle avoidance,'' \emph{ISA transactions}, vol. 133, pp.
  75--90, 2023.

\bibitem{VARGAS2022105054}
S.~Vargas, H.~M. Becerra, and J.-B. Hayet, ``Mpc-based distributed formation
  control of multiple quadcopters with obstacle avoidance and connectivity
  maintenance,'' \emph{Control Engineering Practice}, vol. 121, p. 105054,
  2022.

\bibitem{7782377}
A.~D. Ames, X.~Xu, J.~W. Grizzle, and P.~Tabuada, ``Control barrier function
  based quadratic programs for safety critical systems,'' \emph{IEEE
  Transactions on Automatic Control}, vol.~62, no.~8, pp. 3861--3876, 2017.

\bibitem{COHEN2024100947}
M.~H. Cohen, T.~G. Molnar, and A.~D. Ames, ``Safety-critical control for
  autonomous systems: Control barrier functions via reduced-order models,''
  \emph{Annual Reviews in Control}, vol.~57, p. 100947, 2024.

\bibitem{7857061}
L.~Wang, A.~D. Ames, and M.~Egerstedt, ``Safety barrier certificates for
  collisions-free multirobot systems,'' \emph{IEEE Transactions on Robotics},
  vol.~33, no.~3, pp. 661--674, 2017.

\bibitem{7040372}
A.~D. Ames, J.~W. Grizzle, and P.~Tabuada, ``Control barrier function based
  quadratic programs with application to adaptive cruise control,'' in
  \emph{53rd IEEE Conference on Decision and Control (CDC)}, 2014, pp.
  6271--6278.

\bibitem{10167791}
C.~Jiang and Y.~Guo, ``Incorporating control barrier functions in distributed
  model predictive control for multirobot coordinated control,'' \emph{IEEE
  Transactions on Control of Network Systems}, vol.~11, no.~1, pp. 547--557,
  2024.

\bibitem{9642050}
X.~Tan and D.~V. Dimarogonas, ``Distributed implementation of control barrier
  functions for multi-agent systems,'' \emph{IEEE Control Systems Letters},
  vol.~6, pp. 1879--1884, 2022.

\bibitem{9483029}
J.~Zeng, B.~Zhang, and K.~Sreenath, ``Safety-critical model predictive control
  with discrete-time control barrier function,'' in \emph{2021 American Control
  Conference (ACC)}, 2021, pp. 3882--3889.

\bibitem{9683174}
J.~Zeng, Z.~Li, and K.~Sreenath, ``Enhancing feasibility and safety of
  nonlinear model predictive control with discrete-time control barrier
  functions,'' in \emph{2021 60th IEEE Conference on Decision and Control
  (CDC)}, 2021, pp. 6137--6144.

\bibitem{9516971}
W.~Xiao and C.~Belta, ``High-order control barrier functions,'' \emph{IEEE
  Transactions on Automatic Control}, vol.~67, no.~7, pp. 3655--3662, 2022.

\bibitem{10886244}
P.~Rousseas, D.~Panagou, and K.~Kyriakopoulos, ``Safety-aware trajectory
  tracking using high-order control barrier functions,'' in \emph{2024 IEEE
  63rd Conference on Decision and Control (CDC)}, 2024, pp. 4387--4392.

\bibitem{10155975}
P.~S. Oruganti, P.~Naghizadeh, and Q.~Ahmed, ``Safe control using high-order
  measurement robust control barrier functions,'' in \emph{2023 American
  Control Conference (ACC)}, 2023, pp. 4148--4154.

\bibitem{9868125}
W.~Xiao, C.~Belta, and C.~G. Cassandras, ``Event-triggered control for
  safety-critical systems with unknown dynamics,'' \emph{IEEE Transactions on
  Automatic Control}, vol.~68, no.~7, pp. 4143--4158, 2023.

\bibitem{10156532}
S.~Liu, J.~Zeng, K.~Sreenath, and C.~A. Belta, ``Iterative convex optimization
  for model predictive control with discrete-time high-order control barrier
  functions,'' in \emph{2023 American Control Conference (ACC)}, 2023, pp.
  3368--3375.

\bibitem{1166537}
M.~Sun and D.~Wang, ``Initial shift issues on discrete-time iterative learning
  control with system relative degree,'' \emph{IEEE Transactions on Automatic
  Control}, vol.~48, no.~1, pp. 144--148, 2003.

\bibitem{9777251}
Y.~Xiong, D.-H. Zhai, M.~Tavakoli, and Y.~Xia, ``Discrete-time control barrier
  function: High-order case and adaptive case,'' \emph{IEEE Transactions on
  Cybernetics}, vol.~53, no.~5, pp. 3231--3239, 2023.

\bibitem{10497839}
S.~Zhang, L.~Wang, B.~Xue, D.~Meng, and Q.-G. Wang, ``Consensus criterion
  verification for heterogeneous multiagent systems via sum-of-squares
  programming,'' \emph{IEEE Transactions on Automatic Control}, vol.~69,
  no.~10, pp. 7004--7011, 2024.

\bibitem{10145590}
J.~Fu, G.~Wen, and X.~Yu, ``Safe consensus tracking with guaranteed full state
  and input constraints: A control barrier function-based approach,''
  \emph{IEEE Transactions on Automatic Control}, vol.~68, no.~12, pp.
  8075--8081, 2023.

\bibitem{8039204}
U.~Rosolia and F.~Borrelli, ``Learning model predictive control for iterative
  tasks. a data-driven control framework,'' \emph{IEEE Transactions on
  Automatic Control}, vol.~63, no.~7, pp. 1883--1896, 2018.

\bibitem{9779571}
K.~P. Wabersich and M.~N. Zeilinger, ``Predictive control barrier functions:
  Enhanced safety mechanisms for learning-based control,'' \emph{IEEE
  Transactions on Automatic Control}, vol.~68, no.~5, pp. 2638--2651, 2023.

\end{thebibliography}

\end{document}